%% file: main.tex
\documentclass[11pt]{article}
\usepackage{tikz,tda,authblk,url,hyperref}
\allowdisplaybreaks
\usetikzlibrary{decorations.pathreplacing}
\newcommand{\ggd}{\mathrm{GGD}}
\newcommand{\ged}{\mathrm{GED}}
\newcommand{\cost}{\mathrm{Cost}}
\newcommand{\vol}{\mathrm{Vol}}
\newcommand{\len}[1]{{\lvert\overline{#1}\rvert}}
\usepackage{libertine,newtxmath}
\usepackage[T1]{fontenc}
\title{Distance Measures for Geometric Graphs}
\author[1]{Sushovan Majhi\thanks{smajhi@berkeley.edu}}
\affil[1]{\small School of Information, University of California, Berkeley, USA}
\author[2]{Carola Wenk\thanks{cwenk@tulane.edu; partially supported by NSF grant
CCF 2107434.}}
\affil[2]{Department of Computer Science, Tulane University, New Orleans, USA}
\date{}

\begin{document}
\maketitle\vspace*{-3ex}
\begin{abstract}
A geometric graph is a combinatorial graph, endowed with a geometry that is
inherited from its embedding in a Euclidean space. Formulation of a meaningful
measure of (dis-)similarity in both the combinatorial and geometric structures
of two such geometric graphs is a challenging problem in pattern recognition. We
study two notions of distance measures for geometric graphs, called the
geometric edit distance (GED) and geometric graph distance (GGD). While the
former is based on the idea of editing one graph to transform it into the other
graph, the latter is inspired by inexact matching of the graphs. For decades,
both notions have been lending themselves well as measures of similarity between
attributed graphs. If used without any modification, however, they fail to
provide a meaningful distance measure for geometric graphs---even cease to be a
metric. We have curated their associated cost functions for the context of
geometric graphs. Alongside studying the metric properties of GED and GGD, we
investigate how the two notions compare. We further our understanding of the
computational aspects of GGD by showing that the distance is $\mathcal{NP}$-hard
to compute, even if the graphs are planar and arbitrary cost coefficients are
allowed.
\end{abstract}

\section{Introduction}
Graphs have been a widely accepted object for providing structural
representation of patterns involving relational properties. The framework of
representing complex and repetitive patterns using graphical structures can
facilitate their description, manipulation, and recognition. While hierarchical
patterns are commonly reduced to a string \cite{FU1971155} or a tree
representation \cite{1672247}, non-hierarchical patterns generally require a
graph representation.
One of the most important aspects of such representation is that the problem of
pattern recognition becomes the problem of quantifying (dis-)similarity between
a query graph and a model or prototype graph. The problem of defining a relevant
distance measure for a class of graphs has been looked into for almost five
decades now and has a myriad of applications including chemical structure
matching \cite{willett_similarity_1994}, fingerprint matching
\cite{raymond_effectiveness_2002}, face identification \cite{954601}, and symbol
recognition \cite{954603}. All these applications demand a reliable and
efficient means of comparing two graphs. A meaningful graph distance measure is
expected to yield a small distance implying similarity, and a large distance
revealing disparity.

Depending on the class of graphs of interest and the area of application,
several methods have been proposed. If the use case requires a perfect matching
of two graphs, then the problem of graph isomorphism can be considered
\cite{10.1145/321556.321562}; whereas, subgraph isomorphism can by applied for a
perfect matching of parts of two graphs. These techniques are not, however,
lenient with (sometimes minor) local and structural deformations of the two
graphs. To address this issue, several alternative distance measures have been
studied. We particularly investigate \emph{edit distance}
\cite{sanfeliu_distance_1983,justice_binary_2006} and \emph{inexact matching
distance} \cite{bunke_inexact_1983}. The former makes use of elementary edit
transformations (such as deletion, insertion, relabeling of vertices and edges),
while the latter is based on partially matching two graphs through an inexact
matching relation (\defref{pi}). And, the distance is defined as the minimum
cost of transforming or matching one graph to the other. Although these distance
measures have been battle-proven for attributed graphs (i.e., combinatorial
graphs with finite label sets), the formulations seem inadequate in providing
meaningful similarity measures for geometric graphs. 

A geometric graph  belongs to a special class of attributed graphs having an
embedding into a Euclidean space $\R^d$, where the vertex and edge labels are
inferred from the Euclidean locations of the vertices and Euclidean lengths of
the edges, respectively. In the last decade, there has been a gain in practical
applications involving comparison of geometric graphs. Examples include
road-network or map comparison \cite{akpw-mca-15}, detection of chemical
structures using their spatial bonding geometry, etc. In addition, large
datasets like \cite{da_vitoria_lobo_iam_2008} are being curated by pattern
recognition and machine learning communities.

Despite a rich literature on the matching of attributed graphs and a fair count
of algorithms benchmarked by both the database community and the pattern
recognition community, most of the frameworks become untenable for matching
geometric graphs. They remain oblivious to the spatial geometry such graphs are
endowed with, consequently giving rise to very \emph{artificial} measures of
similarity for geometric graphs. This is not surprising at all---geometric
graphs are a special class of labeled graphs after all! For a geometric graph,
the significant differences include: 
\begin{enumerate}[(i)]
\item Edge relabeling is not an independent edit operation, but vertex labels
dictate the incident edge labels.
\item Vertex relabeling amounts to its translation to a different location in
the ambient space, and additionally incurs the cost of relabeling of all its
adjacent edges.
\end{enumerate}

\subsection{Our Contribution}
We study two distance measures, the \emph{geometric edit distance} ($\ged$) and
\emph{geometric graph distance} ($\ggd$), in order to provide a meaningful
measure of similarity between two geometric graphs. For attributed graphs the
corresponding distance measures are equivalent as shown in \cite[Proposition
1]{bougleux_graph_2017}. In contrast, we show in \secref{compare} they are not
equivalent for geometric graphs. In addition to bounding each distance measure
by a constant factor of the other in \propref{compare}, we provide
polynomial-time computable bounds on them.

We mention here the contribution of \cite{cgkss-msgg-09} for introducing $\ggd$
as well as discussing different definitions of edit distance in the context of
geometric graphs.  The authors also prove certain complexity results for $\ggd$,
which we improve upon in this paper.
%
One of the major contributions of our study is to further our understanding of
the computational complexity of $\ggd$. In \cite{cgkss-msgg-09}, the
authors show that computing $\ggd$ is $\mathcal{NP}$-hard for non-planar graphs,
when arbitrary cost coefficients $C_V,C_E$ (as defined in \defref{ggd}) are
allowed. For planar graphs, $\mathcal{NP}$-hardness is proved under a very
strict condition that $C_V<<C_E$. We show in \propref{hard-2D} that computing
the $\ggd$ is $\mathcal{NP}$-hard, even if the graphs are planar and arbitrary
$C_V,C_E$ are allowed. 

The paper is organized in the following way. In \secref{ged} and \secref{ggd},
we formally define the two distances $\ggd$ and $\ged$, respectively, and
explore some of their important properties. We then compare the two distances in
\secref{compare}. Finally, \secref{complexity} is devoted to our findings on the
computational complexity of the $\ggd$.

\section{Two Distances for Geometric Graphs}\label{sec:distances} A geometric
graph is a combinatorial graph that is also embedded in a Euclidean space. We
begin with the formal definition.
\begin{definition}[Geometric Graph]\label{def:graph} A (finite) combinatorial
graph $G=(V^G,E^G)$ is called a \emph{geometric graph} of $\R^d$ if the vertex
set $V^G\subset\R^d$ and the Euclidean straight-line segments
$\left\{\overline{ab}\mid (a,b)\in E^G\right\}$ intersect (possibly) at their
endpoints.
\end{definition}
We denote the set of all geometric graphs of $\R^d$ by $\G(\R^d)$, and the
subset of geometric graphs without any isolated vertex by $\G_0(\R^d)$. Two
geometric graphs $G=(V^G,E^G)$ and $H=(V^H,E^H)$ are said to be \emph{equal},
written $G=H$, if and only if $V^G=V^H$ and $E^G=E^H$. We make no distinction
between a geometric graph $G=(V^G,E^G)$ and its \emph{geometric realization} as
a subset of $\R^d$; an edge $(u,v)\in E^G$ can be identified as the line-segment
$\overline{uv}$ in $\R^d$, and its length by the Euclidean length $\len{uv}$.
We denote by $\vol(G)$ the sum of the edge lengths of $G$.

\subsection{Geometric Edit Distance (GED)}
\label{sec:ged} Given two geometric graphs $G,H\in\G(\R^d)$, we transform $G$
into $H$  by applying a sequence of edit operations. The allowed edit operations
and their costs are i) inserting (and deleting) a vertex costs nothing, ii)
inserting (and deleting) an edge costs $C_E$ times its length, and iii)
translating a vertex costs $C_V$ times the displacement of the vertex plus $C_E$
times the total change in the length of \emph{all} its incident edges. The
operations and their costs are summarized in \tabref{1}. Throughout the paper,
we assume that the cost coefficients $C_V$ and $C_E$ are positive constants.  In
order to denote a deleted vertex and a deleted edge, we introduce the
\emph{dummy vertex} $\epsilon_V$ and the \emph{dummy edge} $\epsilon_E$,
respectively. While computing edit costs, we follow the convention that
$|\epsilon_E|=0$, $|a-\epsilon_V|=0$ for any $a\in\R^d$, and
$(u,v)=\epsilon_E$ if either $u=\epsilon_V$ or $v=\epsilon_V$. For each
operation $o$ listed in \tabref{1}, note that its inverse, denoted $o^{-1}$, is
also an edit operation with the same cost. 

\begin{table}[h!]
\centering
\begin{tabular}{l | l} 
\hline
Operation & Cost \\ 
\hline
delete (isolated) vertex $u$ & 0 \\
insert vertex $u\in\R^d$ & 0 \\
add edge $e$ between existing vertices & $C_E|e|$ \\
delete edge $e$ & $C_E|e|$ \\
translate a vertex at $u\in\R^d$ to vertex at $v\in\R^d$ & $C_V|u-v| +
\sum\limits_{(s,u)\in E}C_E\big|\len{su}-\len{sv}\big|$\\ 
[0.5ex]
\hline
\end{tabular}
\caption{Allowed edit operations on a geometric graph and associated costs}
\label{table:1}
\end{table}\begin{definition}[Edit Path]\label{def:path} Given two geometric
graphs $G,H\in\G(\R^d)$, an edit path $P$ from $G$ to $H$ is a (finite) sequence
of edit operations $\{o_i\}_{i=1}^k$ that satisfies the following:
\begin{enumerate}[(a)]
    \item $(o_k\circ\ldots\circ o_2\circ o_1)(G)=H$, i.e., $P(G)=H$, and 
    \item $o_{i+1}$ is a legal edit operation on $(o_i\circ\ldots\circ o_2\circ
    o_1)(G)$ for any $1\leq i\leq k-1$.
\end{enumerate}
\end{definition}
Note that we do not require for an intermediate edit operation to yield a
geometric graph. The set of all edit paths between $G,H\in\G(\R^d)$ is denoted
by $\mathcal P(G,H)$. For an edit path $P=\{o_i\}_{i=1}^k$, the edit path
$\{o^{-1}_i\}_{i=0}^k$ from $H$ to $G$ is called its \emph{inverse path}, and is
denoted by $P^{-1}$. For any vertex $u\in V^G$ (resp.~edge $e\in E^G$), we
denote by $P(u)$ (resp.~$P(e)$) the end result after its evolution under $P$. If
$P$ deletes the vertex $u$ (resp.~edge $e$), we write $P(v)=\epsilon_V$
(resp.~$P(e)=\epsilon_E$). The cost, $\cost(P)$, of an edit path $P$ is defined
to be the total cost of the individual edits. 
\begin{definition}[Cost of Edit Paths]\label{def:edit-cost} The cost of an edit
path $P\in\mathcal P(G,H)$, denoted $\cost(P)$, is the sum of the cost of the
individual edits, i.e.,
\[\cost(P)\eqdef\sum_{o_i\in P}\cost(o_i).\]
\end{definition}  
It is not difficult to note that $\cost(P)=\cost(P^{-1})$. Then, $\ged(G,H)$ is
defined as cost of the \emph{least} expensive edit path.
\begin{definition}[Geometric Edit Distance]\label{def:ged} For geometric graphs
$G,H\in\G(\R^d)$, their geometric edit distance, denoted $\ged(G,H)$, is defined
to be the infimum cost of the edit paths, i.e., 
\[
    \ged(G,H)\eqdef\inf_{P\in\mathcal P(G,H)}\cost(P).
\]
\end{definition}
In \propref{ged-metric}, we prove that $\ged$ is, in fact, a metric on the space
of geometric graphs without any isolated vertex. As also observed in
\cite{cgkss-msgg-09}, the following example demonstrates that the distance may
not be attained by an edit path, unless an infinite number of edits are allowed:
Consider $G,H\in\G(\R^2)$, where $G$ has only one edge $(u_1,u_2)$ and $H$ has
only one edge $(v_1,v_2)$ as shown
\input{figures/wiggle.tex}
in \figref{wiggle}. For any fixed $k\geq1$, consider the edit path
$P_k=\{o_i\}_{i=1}^{2k}$, where $o_i$ translates the left vertex of $G$ up by a
distance $1/k$ and then $o_{i+1}$ moves the right vertex by the same distance
for any odd $i$. So, for any $i$
\[
\cost(o_i)=C_V\frac{1}{k}+C_E\left[\sqrt{(1/k)^2+1^2}-1\right]
=C_V\frac{1}{k}+C_E\frac{\frac{1}{k^2}}{\sqrt{\frac{1}{k^2}+1}+1}\mbox{, and therefore}
\]  
\[
\ged(G,H)\leq\cost(P_k)
=\sum_{i=1}^{2k}\cost(o_i)
=2C_V+C_E\frac{\frac{2}{k}}{\sqrt{\frac{1}{k^2}+1}+1}
\xrightarrow{\quad k\to\infty\quad}2C_V.
\]
Now, if we assume that $C_E>C_V$, then any edit path with an edge deletion
costs more than $2C_V$ from \eqnref{split-1}. Therefore, $\ged(G,H)=2C_V$.
However, there is no edit path that attains this cost.

In \defref{edit-cost}, the cost of an edit path $P$ is defined as the aggregated
cost from the individual edits involved in $P$. Another perspective of the cost
of $P$ is the total amount paid by $P$ for the evolution of each vertex and edge
of $G$ and $H$. We make this notion more precise by tracking the evolution of
vertices and edges through their orbit.  
\begin{definition}[Orbit of a Vertex]
Let $P\in\mathcal{P}(G,H)$ be an edit path and $u$ a vertex of $G$. The orbit of
$u$ under $P=\{o_i\}_{i=1}^k$ is the sequence of vertices $\{u_i\}_{i=0}^k$,
where $u_0=u$ and $u_i=(o_i\circ o_{i-1}\circ\ldots\circ o_1)(u)$ for $i\geq1$.
And, the cost of the orbit, denoted $\cost_P(u)$, is defined by
\[ 
    \cost_P(u)\eqdef C_V\sum_{i=1}^k\big|u_i-u_{i-1}\big|.
\]
\end{definition}
The $i$th summand above is positive only if $o_i$ is a translation of the
vertex. Using the triangle inequality, we can immediately note the following
fact.
\input{figures/edits.tex}
\begin{lemma}[Cost of Vertex Orbit]\label{lem:vertex} For a vertex $u\in V^G$
and $P\in\mathcal P(G,H)$, we have
\[ 
    \cost_P(u)\geq C_V|u-P(u)|.
\]
\end{lemma}
We similarly define the orbit of an edge and its cost.
\begin{definition}[Orbit of an Edge]
Let $P\in\mathcal{P}(G,H)$ be an edit path and $e$ an edge of $G$. The orbit of
$e$ under $P=\{o_i\}_{i=1}^k$ is the sequence of edges $\{e_i\}_{i=0}^k$, where
$e_0=e$ and $e_i=(o_i\circ o_{i-1}\circ\ldots\circ o_1)(e)$ for $i\geq1$. And,
the cost of the orbit, denoted $\cost_P(e)$, is defined by
\[ 
\cost_P(e)\eqdef C_E\sum_{i=1}^k\big||e_i|-|e_{i-1}|\big|.
\]
\end{definition}
We note that deletion of the edge or translation of an incident vertex are the
only edit operations in $P$ that can potentially contribute to a positive
summand in the cost function above. Again, the triangle inequality implies the
following lemma.
\begin{lemma}[Cost of Edge Orbit]\label{lem:edge} For an edge $e\in E^G$ and
$P\in\mathcal P(G,H)$, we have
\[ 
    \cost_P(e)\geq C_E\big||e|-|P(e)|\big|.
\]
In particular, $\cost_P(e)\geq|e|$ if $P$ eventually deletes $e$, i.e.,
$P(e)=\epsilon_E$.
\end{lemma}
For examples of vertex and edge orbits see \figref{edit-path}. In order to
describe $\cost(P)$ in terms of the costs of individual orbits, we note that
$\cost(P)$ accounts for the costs of the orbits of:
\begin{enumerate}[(a)]
\item vertices $u\in V^G$ that end up as a vertex of $H$, i.e.,
$P(u)\neq\epsilon_V$
\item vertices $u\in V^G$ with $P(u)=\epsilon_V$
\item vertices $v\in V^H$ that have been inserted, i.e., $P^{-1}(v)=\epsilon_V$
\item edges $e\in E^G$ that end up as an edge of $H$, i.e., $P(e)\neq\epsilon_E$
\item edges $e\in E^G$ with $P(e)=\epsilon_E$
\item edges $f\in E^H$ that have been inserted, i.e.,
$P^{-1}(f)=\epsilon_E$
\item vertices and edges that have been inserted at some point and have also
been deleted eventually.
\end{enumerate}
Moreover, we observe that two vertex (resp. edge) orbits $\{x_i\}$ and $\{y_i\}$
intersect at the $i_0$th position only if $x_i=y_i=\epsilon_V$ (resp.
$x_i=y_i=\epsilon_E$) for all $i\geq i_0$. As a consequence, the positive
summands in the costs of two orbits are necessarily distinct. Accumulating the
costs for all orbits of type (a)--(f), we can, therefore, write 
\begin{equation}\label{eqn:split}
\begin{split}
\cost(P)\geq 
& \underbrace{\sum_{\substack{u\in V^G \\ P(u)\neq\epsilon_V}} \cost_P(u)}
    _\text{vertex translations} + 
\underbrace{\sum_{\substack{u\in V^G \\ P(u)=\epsilon_V}} \cost_P(u)}
    _\text{vertex deletions} + 
\underbrace{\sum_{\substack{v\in V^H \\ P^{-1}(v)=\epsilon_V}} \cost_{P^{-1}}(v)}
    _\text{vertex insertions} \\
& + \underbrace{\sum_{\substack{e\in E^G \\ P(e)\neq\epsilon_E}} \cost_P(e)}
    _\text{edge translations} + 
\underbrace{\sum_{\substack{e\in E^G \\ P(e)=\epsilon_E}} \cost_P(e)}
    _\text{edge deletions} + 
\underbrace{\sum_{\substack{f\in E^H \\ P^{-1}(f)=\epsilon_E}} \cost_{P^{-1}}(f)}
    _\text{edge insertions}.
\end{split}
\end{equation} 
Equation \eqnref{split} together with \lemref{vertex} and \lemref{edge} readily
imply the following useful result.
\begin{lemma}\label{lem:part-1} For any edit path $P\in\mathcal P(G,H)$, it
holds that
\begin{align}\label{eqn:split-1}
\cost(P)\geq \sum_{\substack{u\in V^G \\ P(u)\neq\epsilon_V}}C_V|u-P(u)|
+ \sum_{\substack{e\in E^G \\ P(e)\neq\epsilon_E}}C_E\big||e|-|P(e)|\big| + 
\sum_{\substack{e\in E^G \\ P(e)=\epsilon_E}}C_E|e| 
+ \sum_{\substack{f\in E^H \\ P^{-1}(f)=\epsilon_E}}C_E|f|.
\end{align} 
\end{lemma}

\begin{proposition}[$\ged$ is a Metric]\label{prop:ged-metric} The $\ged$
defines a metric on $\G_0(\R^d)$, the space of geometric graphs without any
isolated vertex. 
\end{proposition}
\begin{proof}
\textbf{Non-negativity.} Since the cost of edit paths are non-negative,
\defref{ged} implies that $\ged(G,H)$ is non-negative for any
$G,H\in\G_0(\R^d)$. \\
    
\noindent\textbf{Separability.} If $\ged(G,H)=0$, we claim that $G=H$, i.e.,
$V^G=V^H$ and $E^G=E^H$. In order to show that $V^G=V^H$, it suffices to show
that the Hausdorff distance $r:=d_H(V^G,V^H)$ between the vertex sets is zero.
Fix 
\[
2\xi=\begin{cases}
C_E\min\{l^G,l^H\},
&\text{ if }r=0 \\
\min\{C_V r, C_E l^G,C_E l^H\},
&\text{ if }r\neq 0
\end{cases}
\]
where $l^G$ and $l^H$ denote the smallest edge lengths of $G$ and $H$,
respectively. Since $\xi>0$, the definition of $\ged$ implies that there is an
edit path $P\in\mathcal P(G,H)$ with $\cost(P)\leq\xi$. Consequently, each of
the four summands in \eqnref{split-1} is no larger than $\xi$. We immediately
see that there is no edge $e\in E^G$ such that $P(e)=\epsilon_E$. Otherwise, the
third summand in \eqnref{split-1} would be at least 
\[ C_E|e|\geq C_E l^G\geq2\xi>\xi, \] leading to a contradiction. The last
inequality above is due to the observation that $\xi>0$. Similarly using the
fourth summand in \eqnref{split-1}, we conclude there is no edge $f\in E^H$ such
that $P^{-1}(f)=\epsilon_E$. In other words, $P$ does not delete any edge of $G$
or $H$, i.e., $|E^G|=|E^H|$. As a result, we can further say that no vertex of
$G$ can be removed and no vertex of $H$ can been inserted, since the input
graphs do not have any isolated vertices. Since $H=P(G)$, the graphs $G$ and $H$
must be isomorphic. Lastly, we show that $V^G=V^H$, i.e., $r=0$. If not, i.e.,
$r\neq0$ and $u_0\in V^G$ such that all the vertices of $H$ are at least $r$
distance away from it, then
\[ C_V|u_0-P(u_0)|\geq C_Vr\geq2\xi>\xi. \] This is a contradiction, because the
first term in \eqnref{split-1} exceeds $\xi$. So, $r=0$. Therefore, $G=H$.\\
    
\noindent\textbf{Symmetry.} Each elementary edit operation can be reversed at
exactly the same cost. Given an edit path $P\in\mathcal{P}(G,H)$, we can reverse
the operations to get an edit path $P^{-1}\in\mathcal{P}(H,G)$ with
$\cost(P)=\cost(P^{-1})$. By \defref{ged}, for an arbitrary $\xi>0$ there exists
$P\in\mathcal{P}(G,H)$ such that $\cost(P)\leq\ged(G,H)+\xi$. On the other hand, 
\[ \ged(H,G)\leq\cost(P^{-1})=\cost(P)\leq\ged(G,H)+\xi. \] Since $\xi$ is
arbitrary, this implies $\ged(H,G)\leq\ged(G,H)$. By a similar argument, one can
also show $\ged(H,G)\geq\ged(G,H)$. Together, they imply
$\ged(H,G)=\ged(G,H)$. \\
    
\noindent\textbf{Triangle Inequality.} Fix an arbitrary $\xi>0$ and
$G,H,I\in\G_0(\R^d)$. By \defref{ged}, there must exist edit paths
$P_1\in\mathcal{P}(G,H)$ and $P_2\in\mathcal{P}(H,I)$ such that
$\cost(P_1)\leq\ged(G,H)+\xi/2$ and $\cost(P_2)\leq\ged(H,I)+\xi/2$. If we
define $P$ to be the concatenation of the edit operations from $P_1$ and $P_2$
in the same order, then $P\in\mathcal P(G,I)$. Moreover,
$\cost(P)=\cost(P_1)+\cost(P_2)$. Now,
\begin{align*}
\ged(G,I) &
\leq\cost(P),\text{ from the Definition of }\ged  \\ & 
=\cost(P_1)+\cost(P_2) \\ &
\leq\left[\ged(G,H)+\frac{\xi}{2}\right]+
\left[\ged(H,I)+\frac{\xi}{2}\right] \\ &
=\ged(G,H)+\ged(H,I)+\xi.
\end{align*}
Since the choice of $\xi$ is arbitrary, we get
$\ged(G,I)\leq\ged(G,H)+\ged(H,I)$.
\end{proof}

\subsection{Geometric Graph Distance (GGD)}\label{sec:ggd} The definition of
$\ged$ is very intuitive but not at all suited for computational purposes.
Firstly, there could be infinitely many locations a vertex is allowed to be
translated to. Secondly, there are infinitely many edit paths between two
graphs---even if the vertices are located on a finite grid. The infinite search
space makes the computation of $\ged$ elusive. As a feasible alternative we
study the $\ggd$. The definition is inspired by the concept of inexact matching
first proposed in \cite{bunke_inexact_1983} for attributed graphs, and later
introduced for geometric graphs in \cite{cgkss-msgg-09}. We follow the notation
of \cite{bunke_inexact_1983} in order to define it. We first define an
\emph{(inexact) matching}.
\begin{definition}[Inexact Matching]\label{def:pi} Let $G,H\in\G(\R^d)$ be two
geometric graphs. A relation
$\pi\subseteq(V^G\cup\{\epsilon_V\})\times(V^H\cup\{\epsilon_V\})$ is called an
(inexact) matching if for any $u\in V^G$ (resp.~$v\in V^H$) there is exactly one
$v\in V^H\cup\{\epsilon_V\}$ (resp.~$u\in V^G\cup\{\epsilon_V\}$) such that
$(u,v)\in\pi$.
\end{definition}
The set of all matchings between graphs $G,H$ is denoted by $\Pi(G,H)$.
Intuitively speaking, a matching $\pi$ is a relation that covers the vertex sets
$V^G,V^H$ exactly once. As a result, when restricted to $V_E$ (resp.~$V^G$), a
matching $\pi$ can be expressed as a map $\pi:V^G\to V^H\cup\{\epsilon_V\}$
(resp.~$\pi^{-1}:V^H\to V^G\cup\{\epsilon_V\}$). In other words, when
$(u,v)\in\pi$ and $u\neq\epsilon_V$ (resp.~$v\neq\epsilon_V$), it is justified
to write $\pi(u)=v$ (resp.~$\pi^{-1}(v)=u$). It is evident from the definition
that the induced map $$\pi:\{u\in V^G\mid\pi(u)\neq\epsilon_V\}\to \{v\in
V^H\mid\pi^{-1}(v)\neq\epsilon_V\}$$ is a bijection. Additionally for edges
$e=(u_1,u_2)\in E^G$ and $f=(v_1,v_2)\in E^H$, we introduce the short-hand
$\pi(e):=(\pi(u_1),\pi(u_2))$ and $\pi^{-1}(f):=(\pi^{-1}(v_1),\pi^{-1}(v_2))$.

Another perspective of $\pi$ is discerned when viewed as a matching between
portions of $G$ and $H$, (possibly) after applying some edits on the two graphs.
For example, $\pi(u)=\epsilon_V$ (resp.~$\pi^{-1}(v)=\epsilon_V$) encodes
deletion of the vertex $u$ from $G$ (resp.~$v$ from $H$), whereas
$\pi(e)=\epsilon_E$ (resp.~$\pi^{-1}(f)=\epsilon_E$) encodes deletion of the
edge $e$ from $G$ (resp.~$f$ from $H$). Once the above deletion operations have
been performed on the graphs, the resulting subgraphs of $G$ and $H$ become
isomorphic, which are finally matched by translating the remaining vertices $u$
to $\pi(u)$. Now, the cost of the matching $\pi$ is defined as the total cost
for all of these operations:
\begin{definition}[Cost of a Matching] Let $G,H\in\mathcal G(\R^d)$ be geometric
graphs and $\pi\in\Pi(G,H)$ an inexact matching. The cost of $\pi$, denoted
$\cost(\pi)$, is defined as
\begin{equation}\label{eqn:split-ggd}
    \cost(\pi)= 
    \underbrace{\sum_{\substack{u\in V^G \\ \pi(u)\neq\epsilon_V}} 
    C_V\lvert u-\pi(u)|}_\text{vertex translations} + 
    \underbrace{\sum_{\substack{e\in E^G \\ \pi(e)\neq\epsilon_E}} 
    C_E\big||e|-|\pi(e)|\big|}_\text{edge translations} + 
    \underbrace{\sum_{\substack{e\in E^G \\ \pi(e)=\epsilon_E}} C_E|e|}
    _\text{edge deletions} + 
    \underbrace{\sum_{\substack{f\in E^H \\ \pi^{-1}(f)=\epsilon_E}} C_E|f|}
    _\text{edge deletions}.
\end{equation} 
\end{definition}

\begin{definition}[$\ggd$]\label{def:ggd} For geometric graphs $G,H\in\G(\R^d)$,
    their geometric graph distance , denoted $\ggd(G,H)$, is defined as the minimum
    cost of an inexact matching, i.e., 
    \[\ggd(G,H)\eqdef\min_{\pi\in\Pi(G,H)}\cost(\pi).\]
\end{definition}
The minimum cost matching between two graphs along with their $\ggd$ has been
illustrated in \figref{wiggle}. The above definition readily yields the
following result.
\begin{lemma}\label{lem:vol} Let $G,H\in\G(\R^d)$ be geometric graphs. For any
$\pi\in\Pi(G,H)$, we have
\[
\cost(\pi)\geq\sum_{\substack{u\in V^G \\ \pi(u)\neq\epsilon_V}} C_V|u-\pi(u)|
+ C_E|\vol(G)-\vol(H)|+2\min\bigg\{
    \sum_{\substack{e\in E^G \\ \pi(e)=\epsilon_E}} C_E|e|, 
    \sum_{\substack{f\in E^H \\ \pi^{-1}(f)=\epsilon_E}} C_E|f|
\bigg\}.
\]
\end{lemma}
\begin{proof}
Without any loss of generality, we assume that
\begin{equation}\label{eqn:assume}
    \sum_{\substack{e\in E^G \\ \pi(e)=\epsilon_E}} C_E|e| \leq 
    \sum_{\substack{f\in E^H \\ \pi^{-1}(f)=\epsilon_E}} C_E|f|.    
\end{equation}
From \eqnref{split-ggd}, we have
\begin{align*}
&\cost(\pi)=\sum_{\substack{u\in V^G \\ \pi(u)\neq\epsilon_V}} C_V|u-\pi(u)|+ 
\sum_{\substack{e\in E^G \\ \pi(e)\neq\epsilon_E}} C_E\big||e|-|\pi(e)|\big| + 
\sum_{\substack{e\in E^G \\ \pi(e)=\epsilon_E}} C_E|e| + 
\sum_{\substack{f\in E^H \\ \pi^{-1}(f)=\epsilon_E}} C_E|f| \\
&=\sum_{\substack{u\in V^G \\ \pi(u)\neq\epsilon_V}} C_V|u-\pi(u)| + 
\sum_{\substack{e\in E^G \\ \pi(e)\neq\epsilon_E}} C_E\big||\pi(e)-|e||\big| + 
\sum_{\substack{f\in E^H \\ \pi^{-1}(f)=\epsilon_E}} C_E|f| - 
\sum_{\substack{e\in E^G \\ \pi(e)=\epsilon_E}} C_E|e| + 
2\sum_{\substack{e\in E^G \\ \pi(e)=\epsilon_E}} C_E|e| \\
&=\sum_{\substack{u\in V^G \\ \pi(u)\neq\epsilon_V}} C_V|u-\pi(u)| + 
\sum_{\substack{e\in E^G \\ \pi(e)\neq\epsilon_E}} C_E\big||\pi(e)-|e||\big| + 
\bigg|\sum_{\substack{f\in E^H \\ \pi^{-1}(f)=\epsilon_E}} C_E|f| - 
\sum_{\substack{e\in E^G \\ \pi(e)=\epsilon_E}} C_E|e|\bigg| \\
&\quad + 2\sum_{\substack{e\in E^G \\ \pi(e)=\epsilon_E}} C_E|e|,
\text{ from \eqnref{assume}} \\
&\geq \sum_{\substack{u\in V^G \\ \pi(u)\neq\epsilon_V}} C_V|u-\pi(u)| + 
\bigg|\sum_{\substack{e\in E^G \\ \pi(e)\neq\epsilon_E}} C_E(|\pi(e)|-|e|)\bigg|
+ \bigg|\sum_{\substack{f\in E^H \\ \pi^{-1}(f)=\epsilon_E}} C_E|f|
-\sum_{\substack{e\in E^G \\ \pi(e)=\epsilon_E}} C_E|e|\bigg|\\
&\quad + 2\sum_{\substack{e\in E^G \\ \pi(e)=\epsilon_E}} C_E|e|,
\text{ by the triangle inequality} \\
&\geq\sum_{\substack{u\in V^G \\ \pi(u)\neq\epsilon_V}} C_V|u-\pi(u)| 
+\bigg|\sum_{\substack{e\in E^G \\ \pi(e)\neq\epsilon_E}} C_E(|\pi(e)|-|e|) 
+\sum_{\substack{f\in E^H \\ \pi^{-1}(f)=\epsilon_E}} C_E|f|
-\sum_{\substack{e\in E^G \\ \pi(e)=\epsilon_E}} C_E|e|
\bigg| \\ 
&\quad + 2\sum_{\substack{e\in E^G \\ \pi(e)=\epsilon_E}} C_E|e|,
\text{ by the triangle inequality} \\
&=\sum_{\substack{u\in V^G \\ \pi(u)\neq\epsilon_V}} C_V|u-\pi(u)| 
+C_E\bigg|\sum_{\substack{e\in E^G \\ \pi(e)\neq\epsilon_E}} |\pi(e)| -  
\sum_{\substack{e\in E^G \\ \pi(e)\neq\epsilon_E}} |e|
+\sum_{\substack{f\in E^H \\ \pi^{-1}(f)=\epsilon_E}} |f| 
-\sum_{\substack{e\in E^G \\ \pi(e)=\epsilon_E}} |e| \bigg| 
+ 2\sum_{\substack{e\in E^G \\ \pi(e)=\epsilon_E}} C_E|e|\\
&=\sum_{\substack{u\in V^G \\ \pi(u)\neq\epsilon_V}} C_V|u-\pi(u)| 
+C_E\bigg|\bigg(\sum_{\substack{e\in E^G \\ \pi(e)\neq\epsilon_E}} |\pi(e)| + 
\sum_{\substack{f\in E^H \\ \pi^{-1}(f)=\epsilon_E}} |f|\bigg)
-\bigg(\sum_{\substack{e\in E^G \\ \pi(e)=\epsilon_E}} |e| 
+\sum_{\substack{e\in E^G \\ \pi(e)\neq\epsilon_E}} |e|\bigg)\bigg| 
+2\sum_{\substack{e\in E^G \\ \pi(e)=\epsilon_E}} C_E|e|\\
&=\sum_{\substack{u\in V^G \\ \pi(u)\neq\epsilon_V}} C_V|u-\pi(u)| 
+C_E\bigg|\sum_{f\in E^H} |f| 
-\sum_{e\in E^G} |e|\bigg| 
+2\sum_{\substack{e\in E^G \\ \pi(e)=\epsilon_E}} C_E|e|\\
&=\sum_{\substack{u\in V^G \\ \pi(u)\neq\epsilon_V}} C_V|u-\pi(u)| 
+C_E\big|\vol(H) - \vol(G)\big|+ 
2\sum_{\substack{e\in E^G \\ \pi(e)=\epsilon_E}} C_E|e|.
\end{align*}
This proves the result.
\end{proof}
The follow proposition provides a lower and upper bound for the $\ggd$
that are computable in polynomial-time.
\begin{proposition}[Bounding the $\ggd$]\label{prop:vol} For geometric graphs
$G,H\in\G(\R^d)$, we have
\[
C_E|\vol(G)-\vol(H)|\leq\ggd(G,H)\leq C_E|\vol(G)+\vol(H)|.
\]
\end{proposition}
\begin{proof}
For any arbitrary matching $\pi\in\Pi(G,H)$, from \lemref{vol} we get
\[
C_E|\vol(G)-\vol(H)|\leq\cost(P).
\]
Since $\pi$ is arbitrary, we conclude $C_E|\vol(G)-\vol(H)|\leq\ggd(G,H)$. 
        
For the second inequality, we choose the trivial matching $\pi_0\in\Pi(G,H)$,
where $\pi_0(u)=\pi_0^{-1}(v)=\epsilon_V$ for all $u\in V^G$ and $v\in V^H$. So,
\[\ggd(G,H)\leq\cost(\pi)=C_E[\vol(G)+\vol(H)].\]
\end{proof}

As also shown in \cite{cgkss-msgg-09},the $\ggd$ is also a metric. We present
a proof here, using our notation, for the sake of completion.
\begin{proposition}[$\ggd$ is a Metric]\label{prop:ggd-metric} The $\ggd$
defines a metric on $\G_0(\R^d)$, the space of geometric graphs without any
isolated vertex. 
\end{proposition}
\begin{proof}
\textbf{Non-negativity.} Since the cost of any matching in $\Pi(G,H)$ is
non-negative, \defref{ggd} implies that $\ggd(G,H)$ is non-negative for any
$G,H\in\G_0(\R^d)$. \\
    
\noindent\textbf{Separability.} If $\ggd(G,H)=0$, then there is $\pi\in\Pi(G,H)$
with $\cost(\pi)=0$. So, all the four summands in \eqnref{split-ggd} are
identically zero. In particular, the third and fourth summands imply that no
edge has been deleted from $G$ or $H$ by $\pi$, i.e., $|E^G|=|E^H|$. Since the
graphs do not have any isolated vertex, this implies that
$\pi(u)\neq\epsilon_V,\pi(v)\neq\epsilon_V$ for all $u\in V^G$ and $v\in V^H$.
As a result, $|E^G|=|E^H|$. Moreover, the first summand of \eqnref{split-ggd}
implies that $\pi(u)=u$ for all $u\in V^G$. Therefore, $G=H$. \\

\noindent\textbf{Symmetry.} We conclude that $\ggd(G,H)=\ggd(H,G)$ due to the
fact that any matching in $\Pi(G,H)$ induces a matching in $\Pi(H,G)$ with
exactly the same cost and vice versa. \\
    
\noindent\textbf{Triangle Inequality.} For the triangle inequality, let us
assume that $\cost(\pi_1)=\ggd(G,H)$ and $\cost(\pi_2)=\ggd(H,I)$ for some
$\pi_1\in\Pi(G,H)$ and $\pi_2\in\Pi(H,I)$. For any $u\in V^G$ and $v\in V^I$,
define $\pi\in\Pi(G,I)$ such that:
\[
\pi(u)=\begin{cases}
\pi_2\circ\pi_1(u),&\text{ if }\pi_1(u)\neq\epsilon_V \\
\epsilon_V,&\text{ otherwise}
\end{cases}
\]
and
\[
\pi^{-1}(v)=\begin{cases}
\pi_1^{-1}\circ\pi_2^{-1}(v),&\text{ if }\pi_2^{-1}(u)\neq\epsilon_V \\
\epsilon_V,&\text{ otherwise}
\end{cases}
\]
Using the triangle inequality, it can be easily seen from \eqnref{split-ggd}
that $\cost(\pi)\leq\cost(\pi_1)+\cost(\pi_2)$. So,
\begin{align*}
\ged(G,I) &\leq\cost(\pi),\text{ from the Definition of }\ged  \\ 
&\leq\cost(\pi_1)+\cost(\pi_2) \\ 
&=\ggd(G,H)+\ggd(H,I).
\end{align*}
Therefore, we get $\ggd(G,I)\leq\ggd(G,H)+\ggd(H,I)$ as desired.
\end{proof}

\subsection{Comparing GED and GGD}\label{sec:compare} As we now have the two
notions of distances under our belts, the question of how they compare arises
naturally. We have already pointed out that the analogous notions for attributed
graphs yield equivalent distances. To our surprise, they are not generally equal
for geometric graphs, as the following proposition demonstrates.
\begin{proposition}\label{prop:tight}
Given any $D>0$, there exist graphs $G,H\in\G(\R)$ such that 
\[\ggd(G,H)=D\text{ and }\ged(G,H)=\left(1+\frac{C_E}{C_V}\right)D.\] In
particular, $\ggd(G,H)<\ged(G,H)$.
\end{proposition}
\begin{proof}
We take two graphs $G,H\in\G(\R)$ as shown in \figref{wiggle-1}. In each graph,
the two vertices are separated by a distance $L$, whereas the second graph is a
copy of the first but shifted by $x$. We also choose
\[ x=\frac{D}{2C_V}\text{ and }L=\left(1+\frac{2C_V}{C_E}\right)x. \] 
\input{figures/wiggle-1.tex}
To see that $\ggd(G,H)=D$, we consider the matching $\pi(u_i)=v_i$ for $i=1,2$.
The cost of the matching is
\[\cost(\pi)=C_V\sum_{i=1}^2|u_i-v_i|=C_V\sum_{i=1}^2x=2C_Vx=D.\] It is worth
noting here that a matching $\pi'$ that is not bijective on the vertex sets has
cost
\[ \cost(\pi')\geq C_EL>C_E\times\frac{2C_V}{C_E}x=D=\cost(\pi). \] Since $L>x$,
the cost of $\pi$ is also (strictly) smaller that $2C_VL$, which is the cost of
the other possible bijective matching. So, we have $\ggd(G,H)=D$.
        
In order to compute $\ged(G,H)$, we consider the edit path $P_0$ that moves the
vertex $u_1$ to $v_1$, then moves $u_2$ to $v_2$. The cost of $P_0$ is 
\[
2C_Vx+2C_Ex=2C_Vx\left(1+\frac{C_E}{C_V}\right)=\left(1+\frac{C_E}{C_V}\right)D.
\]
We now claim that the cost of \emph{any} edit path $P$ is at least
$(1+C_E/C_V)D$. Consider the following two cases:
    
\noindent\textbf{Case I.} If $P(u_1,u_2)=\epsilon_E$, then from
\eqnref{split-1}, we have
\[
\cost(P)\geq2C_EL=2(C_E+2C_V)x=2(C_E+2C_V)\frac{D}{2C_V}=(2+C_E/C_V)D>
(1+C_E/C_V)D.
\]
    
\noindent\textbf{Case II.} For this case, we assume that
$P(u_1,u_2)\neq\epsilon_E$. So, $P$ contains only vertex translations. Let
$O=\{o_i\}_{i=1}^k$ be the subsequence of $P$ containing only those translations
that do not flip the order of the endpoints of the incident edge. Due to the
position of $G$ and $H$, it is evident that $O$ is non-empty. Moreover, the
vertices must travel at least $x$ distance each under $O$. When an endpoint $u$
is moved to a location $w\in\R$ by such an $o_i$, the associated cost of
translating the edge becomes $C_E|w-u|$. Therefore, the cost 
\[
\cost(P)\geq\cost(O)\geq2C_Vx + 2C_E2x=2(C_E+C_V)\frac{D}{2C_V}=(1+C_E/C_V)D.    
\]
Considering the above the cases, we conclude that $\ged(G,H)=(1+C_E/C_V)D$.
\end{proof}

More generally, we prove that following result to compare the two distances. \begin{proposition}\label{prop:compare} For any two geometric graphs
$G,H\in\G(\R^d)$, we have
\[ 
    \ggd(G,H)\leq\ged(G,H)\leq\left(1+\Delta\frac{C_E}{C_V}\right)\ggd(G,H),
\]
where $\Delta$ denotes the maximum degree of the graphs $G,H$.
\end{proposition}
\begin{proof}
Take an arbitrary edit path $P\in\mathcal P(G,H)$. Let us define a matching
$\pi_P\in\Pi(G,H)$ such that 
\[
    \pi_P\eqdef\{(u,P(u))\mid u\in V^G\}\cup\{(P^{-1}(v),v)\mid v\in V^H\}.
\]
This definition of $\pi_P$ implies that $P(u)=\pi_P(u)$ for all $u\in V^G$,
$P(e)=\pi_P(e)$ for all $e\in E^G$, and $P^{-1}(f)=\pi_P^{-1}(f)$ for all $f\in
E^H$. From \eqnref{split-ggd} and \lemref{part-1} it follows that
$\cost(\pi_P)\leq\cost(P)$. The definition of $\ggd(G,H)$ then implies that
\[
    \ggd(G,H)\leq\cost(\pi_P)\leq\cost(P).
\] 
Since $P$ is chosen arbitrarily, the definition of $\ged(G,H)$ then implies the
first inequality.

For the second inequality, we take an arbitrary $\pi\in\Pi(G,H)$. From $\pi$, we
define an edit path $P_\pi$ to be the sequence $(D_E,D_V,T_V,I_V,I_E)$ of edit
operations, where 
\begin{enumerate}[(i)]
\item $D_E$ is a sequence of deletions of edges $e\in E^G$ with
$\pi(e)=\epsilon_E$
\item $D_V$ is a sequence of deletions of vertices $u\in V^G$ with
$\pi(u)=\epsilon_V$,
\item $T_V$ is a sequence of translations of vertices $u\in V^G$ with
$\pi(u)\neq\epsilon_V$ to $\pi(u)$,
\item $I_V$ is a sequence of insertions of vertices $v\in V^H$ with
$\pi^{-1}(v)=\epsilon_V$, and
\item $I_E$ is a sequence of insertions of edges $f\in E^H$ with
$\pi^{-1}(f)=\epsilon_E$.
\end{enumerate}
Each of the above sequences (i)--(v) is unique up to the ordering of the its
operations. Also in $P_\pi$, the edges are deleted in $D_E$ before deleting
their endpoints in $D_V$, and the edges are inserted in $I_E$ only after
inserting their endpoints in $I_V$. Consequently, $P_\pi$ defines a legal edit
path between $G$ and $H$, i.e., $P_\pi\in\mathcal{P}(G,H)$. We claim that
\[
\cost(P_\pi)\leq\left(1+\Delta\frac{C_E}{C_V}\right)\cost(\pi).
\]
To prove the claim, we note that $P_\pi$ does not insert any vertex or edge that
have been later deleted. As a result, the item (g) above \eqnref{split} has a
zero cost. So, \eqnref{split} is, in fact, an equality:
\begin{align*}
\cost(P_\pi)=&
\sum_{\substack{u\in V^G \\ \pi(u)\neq\epsilon_V}} \cost_{P_\pi}(u)
+\sum_{\substack{u\in V^G \\ \pi(u)=\epsilon_V}}\cost_{P_\pi}(u)
+\sum_{\substack{v\in V^H \\ \pi^{-1}(u)=\epsilon_V}}\cost_{P_\pi^{-1}}(v) \\
&+\sum_{\substack{e\in E^G \\ \pi(e)\neq\epsilon_E}}\cost_{P_\pi}(e)
+\sum_{\substack{e\in E^G \\ \pi(e)=\epsilon_E}}\cost_{P_\pi}(e)
+\sum_{\substack{f\in E^H \\ \pi^{-1}(f)=\epsilon_E}}\cost_{P_\pi}(f)    
\end{align*}
Moreover, a deleted (resp. inserted) vertex has never been translated, yielding
a zero cost for its orbit. So, the second and the third summands are identically
zero. We can then write
\begin{align*}
\cost(P_\pi) &=
\sum_{\substack{u\in V^G \\ \pi(u)\neq\epsilon_V}} \cost_{P_\pi}(u)
+\sum_{\substack{e\in E^G \\ \pi(e)\neq\epsilon_E}}\cost_{P_\pi}(e)
+\sum_{\substack{e\in E^G \\ \pi(e)=\epsilon_E}}\cost_{P_\pi}(e)
+\sum_{\substack{f\in E^H \\ \pi^{-1}(f)=\epsilon_E}}\cost_{P_\pi}(f) \\
&=\sum_{\substack{u\in V^G \\ \pi(u)\neq\epsilon_V}} C_V\lvert u-\pi(u)\rvert
+\sum_{\substack{e\in E^G \\ \pi(e)\neq\epsilon_E}} 
\cost_{P_\pi}(e)
+\sum_{\substack{e\in E^G \\ \pi(e)=\epsilon_E}} C_E |e|
+\sum_{\substack{f\in E^H \\ \pi^{-1}(f)=\epsilon_E}}C_E |f| \\
&=\left[\sum_{\substack{u\in V^G \\ \pi(u)\neq\epsilon_V}} C_V\lvert u-\pi(u)\rvert+\sum_{\substack{e\in E^G \\ \pi(e)=\epsilon_E}} C_E |e|
+\sum_{\substack{f\in E^H \\ \pi^{-1}(f)=\epsilon_E}}C_E |f|
\right]+\sum_{\substack{e\in E^G \\ \pi(e)\neq\epsilon_E}} 
\cost_{P_\pi}(e) \\
&\leq\cost(\pi)+\sum_{\substack{e\in E^G \\ \pi(e)\neq\epsilon_E}}\cost_{P_\pi}(e) \\
\end{align*}
In order to get upper bound on the last term, we observe for any edge
$e=(u_1,u_2)\in E^G$ with $\pi(e)\neq\epsilon_E$ that its orbit under $T_V$ is
$\{(u_1, u_2), (u_1, \pi(u_2)), (\pi(u_1), \pi(u_2))\}$. The cost of the orbit
of each $e$ then is
\[
    C_E\left(\big||u_1-\pi(u_2)|-|u_1-u_2|\big|+\big||\pi(u_1)-\pi(u_2)|-|u_1-\pi(u_2)|\big|\right)
    \leq C_E(|u_2-\pi(u_2)|+|u_1-\pi(u_1)|).
\]
So,
\begin{align*}
    \cost(P_\pi) &\leq\cost(\pi)+\sum_{\substack{e\in E^G \\ \pi(e)\neq\epsilon_E}}\cost_{P_\pi}(e) \\
    &\leq\cost(\pi)+\sum_{\substack{e=(u_1,u_2)\in E^G \\ \pi(e)\neq\epsilon_E}}
    C_E(|u_2-\pi(u_2)|+|u_1-\pi(u_1)|)\\
    &\leq\cost(\pi)+\Delta\sum_{\substack{u\in E^V \\ \pi(u)\neq\epsilon_V}}
    C_E|u-\pi(u)| \\
    &\leq\cost(\pi)+\Delta\frac{C_E}{C_V}\sum_{\substack{u\in E^V \\ \pi(u)\neq\epsilon_V}}
    C_V|u-\pi(u)| \\
    &\leq\cost(\pi)+\Delta\frac{C_E}{C_V}\cost(\pi)
    =\left(1+\Delta\frac{C_E}{C_V}\right)\cost(\pi).
\end{align*}
By the definition $\ged$, it is implies that
$\ged(G,H)\leq\left(1+\Delta\frac{C_E}{C_V}\right)\cost(\pi)$. Since $\pi$ is
chosen arbitrarily, we then conclude from the definition of $\ggd$ that 
$\ged(G,H)\leq\left(1+\Delta\frac{C_E}{C_V}\right)\ggd(G,H)$.
\end{proof}
We remark that the configuration in \figref{wiggle} and \propref{tight} show
that the bounds presented in \propref{compare} are, in fact, tight.

\section{Computational Complexity}\label{sec:complexity} In this section, we
discuss the computational aspects of the $\ggd$. The computation is
algorithmically feasible, since the there are only a finite number of matchings
between two graphs. However, it has been already shown in \cite{cgkss-msgg-09}
that the distance is generally hard to compute. We define the decision problem
as follows.
\begin{definition}[PROBLEM $\ggd$]
Given geometric graphs $G,H\in\G(\R^d)$ and $\tau\geq0$, is there a matching
$\pi\in\Pi(G,H)$ such that $\cost(\pi)\leq\tau$?
\end{definition}
In \cite{cgkss-msgg-09}, the authors show that PROBLEM $\ggd$ is
$\mathcal{NP}$-hard for non-planar graphs. For planar graphs, however, its
$\mathcal{NP}$-hardness is proved under the very strict condition that
$C_V<<C_E$. In both cases, the problem instances seem \emph{non-practical}. In
\propref{hard-2D}, we prove a stronger result that the problem is
$\mathcal{NP}$-hard, even if the graphs are planar and arbitrary $C_V,C_E$ are
allowed. Our reduction is from the well-known $3$-PARTITION problem. 
\begin{definition}[Problem $3$-PARTITION]\label{def:3-part} Given positive
integers $N>1$, $B$ and a multiset of positive integers
$S=\{a_1,a_2,\ldots,a_{3N}\}$ so that $\frac{B}{4}<a_i<\frac{B}{2}$ and
$\sum_{i=1}^{3N} a_i=NB$, does there exist a partition of $S$ into $N$ multisets
$S_1,S_2,\ldots,S_N$ such that $|S_i|=3$ and $\sum_{a\in S_i} a=B$ for all
$1\leq i\leq N$?
\end{definition}
The problem is known to be strongly $\mathcal{NP}$-complete
\cite{10.5555/574848}. We reduce an instance $\mathcal{I}:=(N,B,S)$ of
$3$-PARTITION to an instance of PROBLEM $\ggd$.
\begin{proposition}[Hardness of PROBLEM $\ggd$]\label{prop:hard-2D} The PROBLEM
$\ggd$ is $\mathcal{NP}$-hard to decide. This result holds even if
\begin{enumerate}[(i)]
    \item the input graphs are embedded in $\R^2$, and 
    \item the cost coefficients $C_E,C_V$ are arbitrary.
\end{enumerate}
\end{proposition}

\begin{proof}
Given an instance $\mathcal{I}:=(N,B,S)$ of $3$-PARTITION, we construct two
planar graphs $G,H$ such that the existence of a $3$-PARTITION of $S$ implies
$\ggd(G,H)\leq\tau$, otherwise $\ggd(G,H)>\tau$.

We now describe the construction of $G$ and $H$. Each of them will have a
certain number of connected components, which we call \emph{blobs}. A blob of
size $k$ is a connected block of $k$ vertices $\{u_1,u_2,\ldots,u_k\}$ in the
upper row and $k$ vertices $\{l_1,l_2,\ldots,l_k\}$ in the lower row. The two
rows are separated by distance $L$, and the consecutive vertices in each row are
equidistant. The choice of $L$ will be made explicit later on. Except for $u_1$,
each vertex $u_j$ in the upper row is connected to $l_{j-1}$ and $l_j$ in the
bottom row, making the blob path-connected. The configuration of such a typical
blob and its shorthand are depicted in \figref{blob}.
    
\input{figures/blob.tex}
    
We define $G$ as the graph with $3N$ many blobs $G_1,G_2,\ldots,G_{3N}$ of size
$a_1,a_2,\ldots,a_{3N}$, respectively, placed side-by-side so that they do not
overlap. Now, $H$ is defined as the graph with $N$ many blobs
$H_1,H_2,\ldots,H_N$ of size $B$ each placed side-by-side so that they do not
overlap. Now, $G$ and $H$ are placed side-by-side in a bounding-box of width $x$
and height $L$, where 
\[x=\frac{\tau}{2C_V(N+1)NB},\text{ and }L=\frac{\tau}{2C_E(N+1)}.\] We remark
that appropriately small inter-vertex and inter-blob distances can always be
chosen to fit them in the bounding-box, keeping the length of all the vertical
(resp. slanted) edges the same. See \figref{reduction} for the configuration of
the graphs.

\begin{figure}[tbh]
\centering
\input{figures/reduction-1.tex}
\caption{Encoding an instance of $3$-PARTITION into planar graphs
$G,H$}\label{fig:reduction}
\end{figure}
    
Let us first assume that $\mathcal I$ is a YES instance, and that
$\{S_1,S_2,\ldots,S_N\}$ is a partition of $S$. A (bijective) matching
$\pi\in\Pi(G,H)$ can be defined in the following way. For any
$i\in\{1,2,\ldots,N\}$, if $S_i=\{a_{i_1},a_{i_2},a_{i_3}\}$ then the upper and
lower vertices of the blobs $G_{i_1}$, $G_{i_2}$, and $G_{i_3}$ of $G$ are
mapped, consecutively, to the corresponding upper and lower vertices of the
$i$th blob $H_i$ of $H$. We argue that $\cost(\pi)\leq\tau$. In light of
\eqnref{split-ggd}, the cost is the total contribution from the following two
types: 
\begin{enumerate}[(a)]
\item There are $(2NB-N)$ many edges in $G$, whereas there are $(2BN-3N)$ many
in $H$. So, there are exactly $2N$ many vertical edges $e$ in $G$ such that
$\pi(e)=\epsilon_V$. The resulting cost is at most $C_E\cdot2N\cdot L$.
    
\item Since no vertex in the upper row is mapped to a vertex in the lower row
and vice versa, we have
\[|u-\pi(u)|\leq x\text{ for all }u\in G.\] There are $2(\sum_i^{3N} a_i)=2NB$
many vertices in $G$, so the total cost for vertex translation is at most
$C_V\cdot x\cdot 2NB$.
\end{enumerate}
As a result, the total cost is
\[
\cost(\pi)\leq 2C_ENL+2C_VNBx=2C_EN\frac{\tau}{2C_E(N+1)}
+2C_VNB\frac{\tau}{2C_V(N+1)NB}=\tau.
\]
Hence, $\ggd(G,H)\leq\tau$.

For the other direction, we assume that $\ggd(G,H)\leq\tau$, i.e., there is a
matching $\pi\in\Pi(G,H)$ such that $\cost(\pi)\leq\tau$. We observe that
$\pi(V^G)\neq\{\epsilon_V\}$. Otherwise, from \eqnref{split-ggd} the cost of
$\pi$ would be 
\[
\cost(\pi)\geq C_E\vol(G)+C_E\vol(H)\geq C_E(4NB-4N)L
=4C_EN(B-1)\frac{\tau}{2C_E(N+1)}=\frac{2N(B-1)\tau}{N+1}>\tau.
\]
The above volume estimates use the fact that there are $(2NB-3N)$ edges in $G$
and $(2NB-N)$ edges in $H$, and the length of each edge is at least $L$. Also,
the last inequality above is strict because $2N>N+1$ for any $N>1$. Since this
is a contradiction, there must be some $u_0\in V^G$ with
$\pi(u_0)\neq\epsilon_V$. 

Moreover, we claim that $\pi:V^G\to V^H$ must be a bijection. Let us assume the
contrary, i.e., there is $u_1\in V^G$ such that $\pi(u_1)=\epsilon_V$. Since
there is at least one edge (of length at least $L$) incident to $u_1$, we then
have from \lemref{vol},
\begin{align*}
\cost(\pi) &\geq C_V|u_0-\pi(u_0)|+C_E[\vol(H)-\vol(G)]+2C_EL \\
&\geq C_V|u_0-\pi(u_0)|+C_E\times2N\times L+2C_EL \\
&=C_V|u_0-\pi(u_0)|+2C_E(N+1)L \\
&=C_V|u_0-\pi(u_0)|+2C_E(N+1)L\frac{\tau}{2C_E(N+1)} \\
&=C_V|u_0-\pi(u_0)|+\tau.
\end{align*}
Since the graphs are non-overlapping, $u_0-\pi(u_0)>0$. Hence,
$\cost(\pi)>\tau$. This is a contradiction, so $\pi$ must be a bijection.
Finally, we show that $\pi$ defines a partition of $S$ by arguing that a blob
$G_r$ of $G$ cannot split into two blobs $H_s$ and $H_t$ of $H$ when mapped by
$\pi$. If it did, there would an edge $e_0$ of $G_r$ with $\pi(e_0)=\epsilon_E$,
since the blobs $H_s$ and $H_t$ are not connected. This would lead to a
contradiction using the exact same argument just presented. Therefore, $\pi$
defines a partition of the blobs of $G$, so a partition of $S$. This completes
the proof.
\end{proof}
    
\section{Discussions and Future Work}
We have studied two notions for a similarity measure between geometric graphs.
In addition to the metric properties of $\ged$ and $\ggd$, we also establish
tight bounds in order to compare them. Although the distance measures induce
equivalent metrics on the space of geometric graphs, it is not clear which one
is better performant in practical applications. We have also shown the
hardness of computing the $\ggd$ even for planar graphs. This naturally provokes
the question of the hardness of its polynomial-time approximation. We conjecture
that for any $\alpha>1$, an $\alpha$-approximation is also $\mathcal{NP}$-hard,
i.e, PROBLEM $\ggd$ is generally $\mathcal{APX}$-hard. One can also investigate
an alternative version of the $\ged$ that is algorithmically feasible to
compute. This can probably be achieved by putting the graphs on a (Euclidean)
grid and avoiding redundant edit operations in an edit path. It also remains
unclear how to adjust the definitions of the proposed distances to incorporate
merging and splitting of vertices and edges.

\subsection*{Acknowledgments}
The authors thank Erfan Hosseini, Erin Chambers, and Elizabeth Munch for
fruitful discussions.  

{\small\bibliography{references}{}}
\bibliographystyle{plain}

\end{document}

%% file: figures/wiggle.tex
\begin{figure}[tbh]
    \centering
    \begin{subfigure}{0.35\textwidth}
    \centering       
    \begin{tikzpicture}[scale=0.8]
        \draw [very thin, gray] (0,0) node [below] {\footnotesize$(0,0)$} -- (0,3)  node [left] {\footnotesize$(0,1)$}
    -- (3,3) node [right] {\footnotesize$(1,1)$} -- (3,0) node [below]
    {\footnotesize$(1,0)$} -- cycle;
    \filldraw (0,0) circle (2pt) node[anchor=south west] {$u_1$} -- 
    node[auto,swap] {$G$} (3,0) circle
    (2pt) node[anchor=south east] {$u_2$};    
    \filldraw (0,3) circle (2pt) node[anchor=south] {$v_1$} -- 
    node[auto] {$H$} (3,3) circle (2pt) node[anchor=south] {$v_2$};    
    \filldraw[dashed] (0,1) circle (2pt) -- node[auto,sloped] {\tiny$(o_1)(G)$} (3,0);    
    \filldraw[dashed] (0,1) -- node[auto] {\tiny$(o_2\circ o_1)(G)$} (3,1) circle (2pt);    
    \draw[|-|] (-.5,0) -- node[auto] {$1/k$} (-.5,1);
    \node at (1.5,2) {$\vdots$};
    \node at (1.5,2.5) {$\vdots$};
    \end{tikzpicture}
    \end{subfigure}
    \begin{subfigure}{0.35\textwidth}
    \centering
    \end{subfigure}
    \begin{subfigure}{0.35\textwidth}
        \centering       
        \begin{tikzpicture}[scale=0.8]
        \draw [very thin, gray] (0,0) node [below] {\footnotesize$(0,0)$} -- (0,3)  node [left] {\footnotesize$(0,1)$}
        -- (3,3) node [right] {\footnotesize$(1,1)$} -- (3,0) node [below]
        {\footnotesize$(1,0)$} -- cycle;
        \filldraw (0,0) circle (2pt) node[anchor=south west] {$u_1$} -- 
        node[auto,swap] {$G$} (3,0) circle
        (2pt) node[anchor=south east] {$u_2$};    
        \filldraw (0,3) circle (2pt) node[anchor=south] {$v_1$} -- 
        node[auto] {$H$} (3,3) circle (2pt)
        node[anchor=south] {$v_2$};
        \path[->,dashed,bend left,shorten >=0.4cm,shorten <=0.5cm,thick] (0,0) edge node[auto]
        {$\pi(u_1)$} (0,3);     
        \path[->,dashed,bend right,shorten >=0.25cm,shorten <=0.5cm,thick] (3,0) edge node[auto,swap]
        {$\pi(u_2)$} (3,3);     
        \end{tikzpicture}
        \end{subfigure}
        \begin{subfigure}{0.4\textwidth}
        \centering
        \end{subfigure}  
\caption{Left: the edit path $P_k$ alternatively moves the left and right
vertices of $G$ by distance $1/k$. Consequently, $\ged(G,H)=2C_V$. Right: The
inexact matching $\pi$ between $G$ and $H$ has been shown to attain the same
distance for $\ggd(G,H)$.}
    \label{fig:wiggle}
\end{figure}

%% file: figures/edits.tex
\begin{figure}[tbh]
    \centering
    \begin{subfigure}{0.3\textwidth}
    \centering       
    \begin{tikzpicture}
        \foreach \i in {0,...,3} {
            \draw [very thin, gray] (\i,0) -- (\i,3)  node [below] at (\i,0) {\footnotesize$\i$};
        }
        \foreach \i in {0,...,3} {
            \draw [very thin, gray] (0,\i) -- (3,\i) node [left] at (0,\i) {\footnotesize$\i$};
        }
    \filldraw (0,1) circle (2pt) node[anchor=south west] {$u_1$} -- (0,0) circle (2pt) node[anchor=south west] {$u_2$} -- (2,0) circle (2pt) node[anchor=south west] {$u_3$};    
    \end{tikzpicture}
    \end{subfigure}
    \begin{subfigure}{0.3\textwidth}
    \centering
    \begin{tikzpicture}
        \foreach \i in {0,...,3} {
            \draw [very thin, gray] (\i,0) -- (\i,3)  node [below] at (\i,0) {\footnotesize$\i$};
            }
            \foreach \i in {0,...,3} {
                \draw [very thin, gray] (0,\i) -- (3,\i) node [left] at (0,\i) {\footnotesize$\i$};
                }
                \filldraw[red] (0,0) circle (2pt) node[anchor=south west] {$u_2$} -- (2,0) circle
                (2pt) node[anchor=south west] {$u_3$};
                
                \filldraw[blue] (1,2) circle (2pt) node[anchor=south west] {$v_3$} -- (2,0) circle (2pt);
                
                \filldraw[green] (1,2) circle (2pt) -- (3,2) circle (2pt)
                node[anchor=south west] {$v_2$};
    \end{tikzpicture}
    \end{subfigure}
    \begin{subfigure}{0.3\textwidth}
    \centering
    \begin{tikzpicture}
        \foreach \i in {0,...,3} {
            \draw [very thin, gray] (\i,0) -- (\i,3)  node [below] at (\i,0) {\footnotesize$\i$};
        }
        \foreach \i in {0,...,3} {
            \draw [very thin, gray] (0,\i) -- (3,\i) node [left] at (0,\i) {\footnotesize$\i$};
        }
    \filldraw (1,2) circle (2pt) node[anchor=south west] {$v_3$} -- (3,2) circle (2pt) node[anchor=south west] {$v_2$} -- (3,3) circle (2pt) node[anchor=south west] {$v_1$};
    \end{tikzpicture}
    \end{subfigure}
\caption{Two graphs $G,H\in\G(\R^2)$ have been shown on the left and right,
respectively. In the middle, the evolution of $G$ under an edit path
$P=\{o_1,o_2,o_3,o_4\}$ is demonstrated. The edit $o_1$ deletes the edge
$(u_1,u_2)$, then $o_2$ translates $u_2$ to $v_3$, after that $o_4$ translates
$u_3$ to $v_2$, and finally $o_4$ inserts the edge $(v_1,v_2)$. The orbit of the
vertex $u_2$ is $\{u_2,u_2,v_3,v_3,v_3\}$, whereas the orbit of $(u_2,u_3)$ is
$\{(u_2,u_3),(u_2,u_3),(v_3,u_3),(v_3,v_2),(v_3,v_2)\}$.}
\label{fig:edit-path}
\end{figure}

%% file: figures/wiggle-1.tex
\begin{figure}[tbh]
    \centering
    \begin{tikzpicture}[scale=0.8]
\foreach \i in {-1,...,6} { \draw [very thin, gray] (\i,0) -- (\i+1,0); }
\foreach \i in {-1,...,6} { \draw [very thin, gray] (\i,1) -- (\i+1,1); }
\filldraw[thick] (1,0) circle (2pt) node[above] {$v_1$} -- (4,0) circle (2pt)
node[above] {$v_2$}; \filldraw[thick] (0,1) circle (2pt) node[below] {$u_1$} --
(3,1) circle (2pt) node[below] {$u_2$}; \node[above] at (2.5,0) {$H$};
\node[above] at (1.5,1.2) {$G$}; \draw[|<->|] (3,1.5) -- node[auto,swap] {$x$}
(4,1.5);
\draw[|<->|] (1,-.5) -- node[auto,swap] {$L$}
(4,-.5);
\end{tikzpicture}
\caption{The graphs $G$ (top) and $H$ (bottom) are embedded in the real line,
where $u_2-u_1=v_2-v_1=L$ and $v_2-u_2=v_1-u_1=x$.}
\label{fig:wiggle-1}
\end{figure}

%% file: figures/blob.tex
\begin{figure}[thb]
    \centering
    \begin{tikzpicture}[scale=1]
            \draw (-2,0) -- (-2,3);
            \draw (-1,0) -- (-1,3);
            \draw (0,0) -- (-0,3);
            \draw (1,0) -- (1,3);
            \draw (3,0) -- (3,3);
            \draw (-3,3) -- (-2,0);
            \draw (-2,3) -- (-1,0);
            \draw (2,3) -- (3,0);
            \draw (0,3) -- (1,0);
            \draw (3,3) -- (4,0);
            \draw (2,3) -- (2,0);
    
            \fill (-3,3) node[anchor=south] {$u_1$} circle (2pt);
            \fill (-2,0) node[anchor=south west] {$l_1$} circle (2pt); 
            \fill (-2,3) node[anchor=south] {$u_2$} circle (2pt);
            \fill (-1,0) node[anchor=south west] {$l_2$} circle (2pt); 
            \node at (-1,3) {$\ldots$};
            \fill (0,3) node[anchor=south] {$u_j$} circle (2pt);
            \fill (1,0) node[anchor=south west] {$l_j$} circle (2pt); 
            \node at (1,3) {$\ldots$};
            \node at (0,0) {$\ldots$}; 
            \node at (2,0) {$\ldots$}; 
            \fill (2,3) node[anchor=south] {$u_{k-1}$} circle (2pt);
            \fill (3,3) node[anchor=south] {$u_k$} circle (2pt);
            \fill (3,0) node[anchor=south west] {$l_{k-1}$} circle (2pt); 
            \fill (4,0) node[anchor=south west] {$l_k$} circle (2pt);
            \draw[|-|] (5,0) -- node[auto,swap] {$L$} (5,3);

            \begin{scope}[shift={(7.5,0)}]
                \fill[fill=black,draw=black] (0,0) circle (2pt) -- (-1/2,3) circle (2pt) --
                (1/2,3) circle (2pt) -- (1,0) circle (2pt) -- node[auto] {$B$}
                (0,0);
                \draw (0,3) -- node[fill=white] {$\vdots$} (0,0);
                \draw (1/2,3) -- node[fill=white] {$\vdots$} (1/2,0);
                \draw (-1/2,3) -- node[fill=white,inner sep=0] {$\ldots$} (1/2,3);
                \draw (0,0) -- node[fill=white,inner sep=0] {$\ldots$} (1,0);
                \node at (.3,1.35) {$k$};
            \end{scope}
    \end{tikzpicture}            
\caption{Left: A typical blob $B$ of size $k$ is shown. Right: The shorthand
for such a blob is depicted.}\label{fig:blob}
\end{figure}

%% file: figures/reduction-1.tex
\begin{tikzpicture}[scale=1.2]
    \draw[|-|] (-5.5,0) -- node[auto] {$L$} (-5.5,3);
    \draw[|-|] (-5,3.3) -- node[auto] {$x$} (4,3.3);
    \begin{scope}[shift={(-4.5,0)}]
        \fill[fill=black,draw=black] (0,0) circle (2pt) -- (-1/2,3) circle (2pt) --
        (1/2,3) circle (2pt) -- (1,0) circle (2pt) -- node[auto] {$G_1$}
        (0,0);
        \draw (0,3) -- node[fill=white] {$\vdots$} (0,0);
        \draw (1/2,3) -- node[fill=white] {$\vdots$} (1/2,0);
        \draw (-1/2,3) -- node[fill=white,inner sep=0] {$\ldots$} (1/2,3);
        \draw (0,0) -- node[fill=white,inner sep=0] {$\ldots$} (1,0);
        \node at (.3,1.35) {\footnotesize$a_1$};
    \end{scope}
    \node[rotate=10] at (-3.5,1.5) {$\vdots$};
    \begin{scope}[shift={(-3,0)}]
    \fill[fill=black,draw=black] (0,0) circle (2pt) -- (-1/2,3) circle (2pt) --
    (1/2,3) circle (2pt) -- (1,0) circle (2pt) -- node[auto] {$G_i$}
    (0,0);
    \draw (0,3) -- node[fill=white] {$\vdots$} (0,0);
    \draw (1/2,3) -- node[fill=white] {$\vdots$} (1/2,0);
    \draw (-1/2,3) -- node[fill=white,inner sep=0] {$\ldots$} (1/2,3);
    \draw (0,0) -- node[fill=white,inner sep=0] {$\ldots$} (1,0);
    \node at (.3,1.35) {\footnotesize$a_i$};
    \end{scope}
    \node[rotate=10] at (-2,1.5) {$\vdots$};
    \begin{scope}[shift={(-1.5,0)}]
        \fill[fill=black,draw=black] (0,0) circle (2pt) -- (-1/2,3) circle (2pt) --
        (1/2,3) circle (2pt) -- (1,0) circle (2pt) -- node[auto] {$G_{3N}$}
        (0,0);
        \draw (0,3) -- node[fill=white] {$\vdots$} (0,0);
        \draw (1/2,3) -- node[fill=white] {$\vdots$} (1/2,0);
        \draw (-1/2,3) -- node[fill=white,inner sep=0] {$\ldots$} (1/2,3);
        \draw (0,0) -- node[fill=white,inner sep=0] {$\ldots$} (1,0);
        \node at (.27,1.35) {\footnotesize$a_{3N}$};
    \end{scope}
    \draw [decorate,decoration={brace,amplitude=5pt,mirror,raise=4ex}]
    (-4.5,0) -- (-0.5,0) node[midway,yshift=-3em]{$G$};
\begin{scope}[shift={(0,0)}]
    \fill[fill=black,draw=black] (0,0) circle (2pt) -- (-1/2,3) circle (2pt) --
    (1/2,3) circle (2pt) -- (1,0) circle (2pt) -- node[auto] {$H_1$}
    (0,0);
    \draw (0,3) -- node[fill=white] {$\vdots$} (0,0);
    \draw (1/2,3) -- node[fill=white] {$\vdots$} (1/2,0);
    \draw (-1/2,3) -- node[fill=white,inner sep=0] {$\ldots$} (1/2,3);
    \draw (0,0) -- node[fill=white,inner sep=0] {$\ldots$} (1,0);
    \node at (.3,1.35) {\footnotesize$B$};
\end{scope}
\node[rotate=10] at (1,1.5) {$\vdots$};
\begin{scope}[shift={(1.5,0)}]
        \fill[fill=black,draw=black] (0,0) circle (2pt) -- (-1/2,3) circle (2pt) --
        (1/2,3) circle (2pt) -- (1,0) circle (2pt) -- node[auto] {$H_j$}
        (0,0);
        \draw (0,3) -- node[fill=white] {$\vdots$} (0,0);
        \draw (1/2,3) -- node[fill=white] {$\vdots$} (1/2,0);
        \draw (-1/2,3) -- node[fill=white,inner sep=0] {$\ldots$} (1/2,3);
        \draw (0,0) -- node[fill=white,inner sep=0] {$\ldots$} (1,0);
        \node at (.3,1.35) {\footnotesize$B$};
\end{scope}
\node[rotate=10] at (2.5,1.5) {$\vdots$};
\begin{scope}[shift={(3,0)}]
    \fill[fill=black,draw=black] (0,0) circle (2pt) -- (-1/2,3) circle (2pt) --
    (1/2,3) circle (2pt) -- (1,0) circle (2pt) -- node[auto] {$H_N$}
    (0,0);
    \draw (0,3) -- node[fill=white] {$\vdots$} (0,0);
    \draw (1/2,3) -- node[fill=white] {$\vdots$} (1/2,0);
    \draw (-1/2,3) -- node[fill=white,inner sep=0] {$\ldots$} (1/2,3);
    \draw (0,0) -- node[fill=white,inner sep=0] {$\ldots$} (1,0);
    \node at (.3,1.35) {\footnotesize$B$};
\end{scope}
\draw [decorate,decoration={brace,amplitude=5pt,mirror,raise=4ex}]
(0,0) -- (4,0) node[midway,yshift=-3em]{$H$};
\end{tikzpicture}